\documentclass[doublecolumn]{IEEEtran}
\usepackage{calc}
\usepackage[T1]{fontenc}

\usepackage[overload]{empheq}
\usepackage{cleveref}
\usepackage{multirow}
\usepackage{cite}
\usepackage{graphicx,subfigure}
\usepackage{psfrag}
\usepackage{amsmath,amssymb}
\usepackage{color}
\usepackage{tikz}
\usepackage{cleveref}
\usepackage{empheq}
\usepackage[nopar]{lipsum}

\newenvironment{proof}{\begin{IEEEproof}}{\end{IEEEproof}}

\newtheorem{proposition}{Proposition}
 
\newcommand{\bit}{\begin{itemize}}
\newcommand{\eit}{\end{itemize}}

\newcommand{\bc}{\begin{center}}
\newcommand{\ec}{\end{center}}

\newcommand{\ba}{\begin{array}}
\newcommand{\ea}{\end{array}}

\newcommand{\beq}{\begin{equation}}
\newcommand{\eeq}{\end{equation}}

\newcommand{\beqn}{\begin{equation*}}
\newcommand{\eeqn}{\end{equation*}}

\newcommand{\bean}{\begin{eqnarray*}}
\newcommand{\eean}{\end{eqnarray*}}
\newcommand{\bea}{\begin{eqnarray}}
\newcommand{\eea}{\end{eqnarray}}

\def\F{\mathbb{F}}

\def\cv{\boldsymbol{c}}

\def\sv{\boldsymbol{s}}

\def\wv{\boldsymbol{w}}
\def\xv{\boldsymbol{x}}
\def\yv{\boldsymbol{y}}

\makeatletter
\def\blfootnote{\gdef\@thefnmark{}\@footnotetext}
\makeatother

\begin{document}
\sloppy

\title{On Model Coding for Distributed Inference and Transmission in Mobile Edge Computing Systems}
\author{
  Jingjing Zhang and Osvaldo Simeone
	}

\maketitle

\thispagestyle{empty}

\begin{abstract}
Consider a mobile edge computing system in which users wish to obtain the result of a linear inference operation\blfootnote{The authors are with the Department of Informatics at King's College London, UK (emails: jingjing.1.zhang@kcl.ac.uk, osvaldo.simeone@kcl.ac.uk). The authors have received funding from the European Research Council (ERC) under the European Union's Horizon 2020 Research and Innovation Programme (Grant Agreement No. 725731).} on locally measured input data. Unlike the offloaded input data, the model weight matrix is distributed across wireless Edge Nodes (ENs). ENs have non-deterministic computing times, and they can transmit any shared computed output back to the users cooperatively. This letter investigates the potential advantages obtained by coding model information prior to ENs' storage. Through an information-theoretic analysis, it is concluded that, while generally limiting cooperation opportunities, coding is instrumental in reducing the overall computation-plus-communication latency.
\end{abstract}

\section{Introduction} 
Introduced by the European Telecommunications Standards Institute (ETSI), the concept of mobile edge computing is by now established as a pillar of the 5G network architecture as an enabler of computation-intensive applications on mobile devices \cite{TSMFDS:17}. As illustrated in Fig.~\ref{fig:model}, with mobile edge computing, users offload local data to edge servers connected to wireless Edge Nodes (ENs). These in turn carry out the necessary computations and return the desired output to the users on the wireless downlink. Most academic work on mobile edge computing has focused on the complex resource allocation problem of orchestrating computing and communication resources at the mobiles and at the ENs (see, e.g., \cite{SSB:15} and references therein).

Papers in the line of work introduced above either assume generic applications characterized by given input-output rate requirements (e.g., \cite{SSB:15}) or optimize the partition of the computing graph of the applications between local and edge computing. Moreover, this body of research has shown the importance of jointly designing the physical-layer transmission strategy and the computing schedule. Importantly, computing the same output at multiple ENs, while generally increasing the computation time, enables cooperation opportunities in the downlink transmission from the ENs to the users \cite{SSB:15}.

More recently, in a parallel development in the information-theoretic literature, it has been demonstrated that, if the computation of interest has specific properties, coding of either inputs or outputs can help decrease the overall latency. In particular, 
reference \cite{LLPPR:18} demonstrated the advantages of Maximum Distance Separable (MDS) coding of input matrices in reducing the latency for distributed matrix-vector multiplication in master-worker systems. The impact of coding computational outputs was instead investigated in \cite{LMA:27} for Map-Reduce computing tasks.  


In this letter, we investigate the role of coding in the mobile edge computing system illustrated in Fig.~\ref{fig:model}. 
In the system, each user wishes to compute a linear inference $\mathbf{W}\xv$ on a local data vector $\xv$ given a network-side model matrix $\mathbf{W}$ via offloading. The matrix $\mathbf{W}$ is generally large and hence it requires splitting across the servers of multiple ENs. 
Linear operations are practically important, e.g., for the implementation of recommendation systems based on collaborative filtering \cite{SFHS:07} or similarity search based on the cosine distance \cite{RYR:07}. In both cases, the user-side data is a vector $\xv$ that embeds the user profile \cite{SFHS:07} or a query \cite{RYR:07}, and the goal is to search through the matrix of all items on the basis of the inner products between the corresponding row of matrix $\mathbf{W}$ and the user-data $\xv$. This letter presents an information-theoretic framework that enables the potential advantages of model coding and associated performance trade-offs to be quantified.


\begin{figure}[t!] 
  \centering
\includegraphics[width=1\columnwidth]{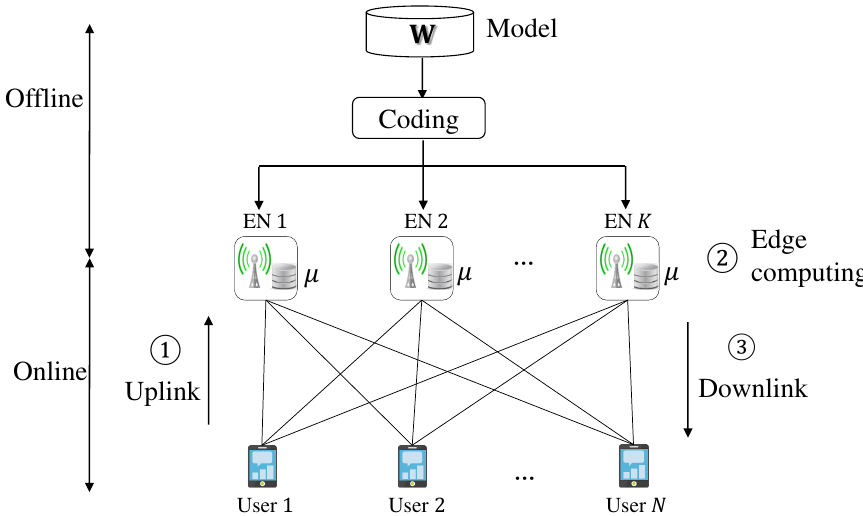}
\vspace{-21pt}
\caption{Illustration of the distributed edge computing system under study.}
\label{fig:model}
\vspace{-15pt}
\end{figure}

\section{System Model and Performance Criteria}

\subsection{System Model}

We consider the distributed edge computing model illustrated in Fig.~\ref{fig:model}, where $N$ users are connected to $K$ ENs through a shared wireless channel. For a given input vector $\xv\in \F_{2^L}^{r \times 1}$ of $rL$ bits provided by a user, the system aims at computing the linear inference operation $\yv=\mathbf{W}\xv$, where the weight, or model, matrix $\mathbf{W}\in \F_{2^L}^{m \times r}$ is static for a sufficiently long period of time. Each EN $k$ can store a number of bits equivalent to a fraction $\mu\in[1/K,1]$ of rows of matrix $\mathbf{W}$, i.e., $m \mu rL$ bits. Storage of information from matrix $\mathbf{W}$ takes place offline given the static nature of the model. 


Each user $n$, with $n\in[N]$, has its own personal data $\xv_n$, with $\xv_n\in \F_{2^L}^{r \times 1}$ of $rL$ bits, which is collected online by the user, and it wishes to obtain the result of the linear operation $\yv_n=\mathbf{W}\xv_n$. The task is offloaded to the ENs as shown in Fig.~\ref{fig:model}. To this end, the ENs acquire the user data $\mathbf{X}=[\xv_1,\cdots,\xv_N]$ through uplink transmission. Second, the ENs carry out computations on the received users' data and on the stored data about $\mathbf{W}$. Finally, via downlink communication, the ENs deliver the results of the computations to the users, so that each user $n$ can recover the required output $\yv_n$. 

In this letter, we make the simplifying assumption that the time needed to upload $\mathbf{X}$ to all ENs is fixed and each EN gets the entire matrix $\mathbf{X}$. This allows us to focus on the challenging problem of jointly designing offline model coding and storage at the ENs, as well as online edge computing and downlink transmission phases. The problem is formulated as follows. 

\emph{Model Coding and Storage:} In an \emph{offline} phase, the model matrix $\mathbf{W}$ is linearly encoded \cite{ZS:18} as 
$[\cv_1^T, \cdots,\cv_{m'}^T]^T=\mathbf{G}\mathbf{W}, $
where we have defined the \emph{coding matrix} $\mathbf{G}\in \F_{2^L}^{m' \times m}$, with integer $m'\geq m$. Each EN $k$ stores the subset $\mathcal{C}_k$, with $\mathcal{C}_k\subseteq \mathcal{C}$ of $|\mathcal{C}_k|\leq m\mu$ coded rows.

\emph{Edge Computing:} In the \emph{online} phase, each EN $k$ computes inner products between all users' data received in the uplink and the available coded model rows in set $\mathcal{C}_k$. As in \cite{OUG:18}, the order in which such computations are carried out is specified by vector $\mathbf{s}^T_k=[\sv_{1,k},\cdots,\sv_{m\mu,k}]$, where each element $s_{i,k}\in \mathbb{F}_{2^L}^{1 \times r}$, with $i\in[m\mu]$, is selected from the set $\mathcal{C}_k$ of coded rows available at EN $k$. In particular, each EN $k$ starts to compute the inner product $\sv_{1,k}\mathbf{X}$ and continues computing $\sv_{i,k} \mathbf{X}\in \F_{2^L}^{1 \times N}$, for $i=2,3,\cdots,m\mu$. As in the literature on distributed computing, we refer to each computation $\sv_{i,k} \mathbf{X}$ as an Intermediate Value (IV) \cite{LMA:16}. A computation policy is hence defined by the coding matrix $\mathbf{G}$, \emph{scheduling matrix} $\mathbf{S}\in \F_{2^{rL}}^{m\mu\times K}$, with the $k$th column vector given as $\mathbf{s}_k$, as well as by a \emph{stopping criterion}, which is used by the ENs to decide when to stop the computing phase and start downlink transmission. 

To formulate the stopping criterion, we define $\mathbf{m}(t)=[m_1(t),\cdots,m_K(t)]$ as the vector that indicates how many IVs have been computed by the ENs by time $t$, with $t=0$ indicating the start of the computing phase and $m_k(t)$ denoting the number of computations at each EN $k$. Note that we have the inequalities $0\leq m_k(t)\leq m\mu$ due to the storage constraint. We also define as 
\begin{align} \label{im}
\mathcal{I}_k(m_k,\mathbf{s}_k)=\{\sv_{i,k} \mathbf{X}: i\in[m_k] \}, 
\end{align} 
the set of first $m_k$ IVs computed by EN $k$ for a given choice of the scheduling vector $\mathbf{s}_k$. A computation vector $\mathbf{m}$ is said to be feasible if the union $\bigcup_{k\in[K]}\mathcal{I}_k(m_k,\mathbf{s}_k)$ of all computed IVs across all $K$ ENs contains enough information to enable the recovery of all the outputs $\{\yv_n\}_{n=1}^{N}$, i.e., if the conditional entropy $H\big(\{\yv_n\}_{n=1}^{N}|\bigcup_{k\in[K]}\mathcal{I}_k(m_k,\mathbf{s}_k)\big)$ equals zero. Note that, if $\mathbf{m}$ is feasible, then any $\mathbf{m'}\geq \mathbf{m}$, where inequality is element-wise, is also feasible. 

A stopping criterion for a given computation policy is defined by a set $\mathcal{M}$ of feasible computation vectors in the sense that the ENs stop computing at the first time $T_C$ such that $\mathbf{m}(T_C)$ is in set $\mathcal{M}$ , i.e., 
\begin{align} \label{def:time}
T_C=\min\{t:\mathbf{m}(t)\in \mathcal{M}\}.
\end{align} 
As a result, the computed IVs at EN $k$ by the end of the edge computing phase are given as $\mathcal{I}_k=\mathcal{I}_k(m_k(T_C),\mathbf{s}_k)$. As a simple example, a computation policy may require that all ENs complete all local computations, i.e., $\mathcal{M}=\{[m\mu,m\mu,\cdots,m\mu]\}$.

\emph{Downlink Communication:} In this phase, the ENs send the computed IVs to the users on the downlink so that each user $n$ can recover the desired output $\yv_n$.
To this end, the ENs apply conventional one-shot linear precoding as in \cite{ZS:17,NMA:17}. Accordingly, in each downlink transmission block, the transmitted signal at each EN $k\in [K]$ is given as $u_k = a_k s_{k}$, where $s_{k}$ is a symbol that encodes a subset of IVs in set $\mathcal{I}_k$, and $a_{k}$ is the corresponding beamforming coefficients. All the ENs that have computed the same IVs can transmit them cooperatively via joint beamforming \cite{ZS:17,NMA:17}. We impose the per-EN power constraint $\mathbb{E} \left[|u_k|^2\right] \leq P$. In each downlink block, the signal received by each user $n$ is given as
\begin{equation} \label{trans}
    v_n=\sum_{k=1}^{K} h_{nk} u_k + z_n,
\end{equation}
where $h_{nk} \in \mathbb{C}$ is the channel coefficient from EN $k$ to user $n$; $u_k \in \mathbb{C}$ is  the defined signal transmitted by EN $k$; $z_n$ is unit-power additive complex Gaussian noise. The fading channels are drawn from a continuous distribution, constant in each block, and known to all ENs.
\vspace{-5pt}

\subsection{Performance Analysis}

As in \cite{MCJ:18}, we assume that the computing time needed by each EN $k$ to perform $m_k$ computations is given as 
\begin{align} \label{eq:time}
t_k=\lambda_k +\tau m_k,
\end{align} 
where  $\lambda_k \sim \text{exp}(\eta)$, independent across ENs, is an exponential random variable with average $1/\eta$ that models the time needed for setup at each EN $k$; and $\tau$ is the (deterministic) time required for each computation. Under model \eqref{eq:time}, given a stopping set $\mathcal{M}$, the random duration $T_C$ in \eqref{def:time} of the computation phase can be written as the optimization
\begin{align} \label{def:time2}
T_C=\max_{k\in[K]}\big(\lambda_k+\tau m^*_k(\pmb{\lambda})\big),
\end{align} 
where we have defined the stopping vector $\mathbf{m}^*(\pmb{\lambda})=[m^*_1(\pmb{\lambda}),\cdots,m^*_K(\pmb{\lambda})]$ for a given vector $\pmb{\lambda}=[\lambda_1,\cdots,\lambda_K]$ as 
\begin{align} \label{def:mstar}
\mathbf{m}^*(\pmb{\lambda})=\text{arg} \min_{ \mathbf{m} \in \mathcal{M}} \max_{k\in[K]}(\lambda_k+\tau m_k).
\end{align}
This follows since the time needed to realize a computation vector $\mathbf{m}$ is given by $\max_{k\in[K]}(\lambda_k+\tau m_k)$. 


In the high-SNR regime of interest, we evaluate the downlink phase duration $T_D$ by normalizing for the time $NL/\log(P)$ needed to deliver one IV, of size $NL$ bits, to all $N$ users, in the absence of mutual interference. Hence, the normalized communication delay $\delta_D$ is given as 
\begin{equation}  \label{eq:sd}
\delta_D = \lim_{P \rightarrow \infty}\frac{T_D}{NL/\log(P)}.
\end{equation}
For comparison, we also normalize the computation time $T_C$ by the time $\tau$ to compute one IV for all users, obtaining the normalized computation delay $\delta_C=T_C/\tau$. Finally, the average total normalized latency $\delta$ of the edge computing system is given as 
\begin{align} \label{def:total}
\delta=\textrm{E}[\delta_C]+\gamma\textrm{E}[\delta_D],
\end{align}
where parameter $\gamma$ is the ratio between the average time (in seconds) needed to compute one IV at an EN and the average time needed to transmit one IV on an interference-free channel.

\section{Uncoded vs. Coded Computing}

\subsection{Uncoded Storage and Computing (UC)} \label{uc}
Consider first a standard uncoded strategy whereby each EN stores $m\mu$ rows directly from the model matrix rows $\{\wv_i\}_{i=1}^{m}$. Following, e.g., \cite{OUG:18}, the scheduling matrix $\mathbf{S}$ is designed in a cyclic manner, so that each vector $\wv_i$ is repeated $K\mu$ times across all ENs. 
As an example, if $m=6$, $\mu=1/2$ and $K=3$, then the scheduling vector are $\mathbf{s}_1=[\wv_1, \wv_4,\wv_5]$, $\mathbf{s}_2=[\wv_2, \wv_5,\wv_6]$, and $\mathbf{s}_3=[\wv_3, \wv_6,\wv_4]$. 
The stopping set $\mathcal{M}$ is defined as the set of all feasible computation vectors, so that every vector $\mathbf{m}\in \mathcal{M}$ ensures that each IV $\wv_i\mathbf{X}$ has been computed by some EN. 

For each IV $\wv_i\mathbf{X}$ and a given feasible vector $\mathbf{m}\in \mathcal{M}$, we define as $r_i(\mathbf{m})$ the number of times that the IV has been computed across the ENs, i.e., the number of ENs whose set $\mathcal{I}_k$ contains the IV. We hence have the constraint $\sum_{i=1}^{m}r_i(\mathbf{m})=\sum_{k=1}^{K} {m_k}$. To deliver a single IV computed at $r_i(\mathbf{m})$ ENs, cooperative Zero-Forcing (ZF) precoding allows $\min\{r_i(\mathbf{m}),N\}$ users to be served at the same time at the maximum high-SNR rate $\log(P)$, where $\min\{a,b\}$ represents the minimum between the two arguments $a$ and $b$. This is done by choosing the precoding matrix across the $\min\{r_i(\mathbf{m}),N\}$ transmitting ENs to equal the inverse of the (square) channel matrix, upon appropriate power scaling. Hence, the normalized downlink latency \eqref{eq:sd} for this IV is given as $1/\min\{r_i(\mathbf{m}),N\}$ \cite{ZS:17,NMA:17}. As a result, the total latency can be characterized as follows.

\begin{proposition}
With the described uncoded strategy, the average total normalized latency \eqref{def:total} is given as 
\begin{align}\label{latencyuc} 
\!\!\delta_{UC}\!=\!\textrm{E}\Bigg[\!\frac{\max_{k\in[K]}\!\!\big(\lambda_k\!+\!\tau m^*_k(\pmb{\lambda})\big)}{\tau}\!+\!\!\!\sum_{i\in[m]}\!\!\frac{\gamma}{\min\{r_i(\mathbf{m^*(\pmb{\lambda})}),N\}}\!\Bigg],   
\end{align}
where the stopping vector $\mathbf{m}^*(\pmb{\lambda})$ is given in \eqref{def:mstar}, and the expectation is taken over the distribution of the random vector $\pmb{\lambda}$. 
\end{proposition}
\vspace{-5pt}

\subsection{MDS coded Storage and Computing (MC)}
We proceed to consider an MDS-coded scheme that aims at enhancing robustness to straggling ENs \cite{LMA:16,MCJ:18,ZS:18}. In this scheme, the coding matrix $\mathbf{G}$ is selected as the generator matrix of an $(K\mu m,m)$ MDS code; each EN $k$ stores $m\mu$ distinct coded rows; and the computing order at each EN is arbitrary. Furthermore, the stopping set $\mathcal{M}$ is defined such that, given the fractional cache size $\mu$, the system waits for the fastest $\lceil 1/\mu \rceil$ ENs to finish all their computations. By definition of an $(K\mu m,m)$ MDS code, this guarantees that all the $m$ required output elements in $\{\yv_n\}_{n=1}^{N}$ can be obtained from the $m$ IVs computed at the $[1/\mu]$ ENs by treating the missing IVs from the slower $K-\lceil 1/\mu \rceil$ ENs as erasures. 


With this scheme, there is no redundancy in the set of IVs computed at the ENs and hence no cooperation opportunities are available for downlink transmission. It follows that the $m$ IVs need to be sent sequentially to each user in the downlink using orthogonal transmission, and thus the communication latency is given as $\delta_D=m$.

\begin{proposition} \label{pro:mds}
With the described MDS coded scheme, the average total latency \eqref{def:total} is given as 
\begin{align}  \label{latencymc}
\delta_{MC}=\frac{(H_K-H_{K-\lceil 1/\mu \rceil})}{\eta\tau}+m(\mu+\gamma). 
\end{align}
\end{proposition}

\begin{proof}
Since only the fastest $\lceil 1/\mu \rceil$ ENs are required to execute their full computations, the average computation time is given as $\textrm{E}[T_C]=\textrm{E}[\lambda_{\lceil 1/\mu \rceil:K}]+\tau m\mu=(H_K-H_{K-\lceil 1/\mu \rceil})/\eta+\tau m\mu$, where $\lambda_{\lceil 1/\mu \rceil:K}$ is the $\lceil 1/\mu \rceil_{th}$ order statistics of exponential random variables $\{\lambda_k\}_{k=1}^{K}$, and $H_K=\sum_{k=1}^{K}1/k$ is the $K_{th}$ harmonic number (see \cite{MCJ:18}). 
\end{proof}
\vspace{-5pt}

\subsection{Hybrid Scheme (HS)} 
We now propose a hybrid scheme whose aim is to combine the robustness to stragglers afforded by the MDS-coded scheme and the cooperative downlink transmission advantages of the uncoded scheme. The proposed hybrid scheme allows the reduction in computing time via MDS coding to be traded off for savings in communication time via EN cooperation.
To this end, we concatenate an $(\rho_1m,m)$ MDS code for some $\rho_1\geq 1$ with a repetition code that replicates each coded vector to $\rho_2$ ENs. Controlling the design parameters $(\rho_1,\rho_2)$, the scheme ranges from uncoded storage $(\rho_1=1)$ to MDS coding $(\rho_2=1)$.

More precisely, following \cite{ZS:18}, in order to ensure an even distribution of coded rows, the $\rho_1m$ coded rows $\{\cv_i\}_{i=1}^{\rho_1m}$ are split into $\binom{K}{\rho_2}$ disjoint subsets. Each subset $\mathcal{C}_{\mathcal{K}}$ consists of $b=(\rho_1m)/\binom{K}{\rho_2}$ coded rows, and is indexed by a subset $\mathcal{K}\subseteq[K]$ of size $\rho_2$, i.e., $|\mathcal{K}|=\rho_2$. Each EN $k$ stores all the rows in the set $\bigcup_{\mathcal{K}: k\in\mathcal{K}} \mathcal{C}_{\mathcal{K}}$, with cardinality $b\binom{K-1}{\rho_2-1}=\rho_1\rho_2m/K$. Due to the storage constraint $m\mu$ at each EN, we have the constraint 
\begin{align} \label{condition1}
\rho_1\rho_2\leq K\mu.
\end{align}

We select the stopping set in a manner similar to the MDS coded strategy, so that the computing phase is completed as soon as $q$ ENs complete all their computations, where $q$ is a design parameter. Following \cite[Proposition 1]{ZS:18}, the three design parameters $(q, \rho_1, \rho_2)$ need to satisfy the constraint 
\begin{align} \label{condition2}
\binom{K}{\rho_2}-\binom{K-q}{\rho_2} \geq\frac{1}{\rho_1} \binom{K}{\rho_2}
\end{align}
in order to ensure that $m$ distinct coded IVs are computed across the ENs and hence all desired outputs can be recovered. It can be observed that the choice of parameters $(\rho_1,\rho_2)$ depends on system parameters $K, \mu$ and $\gamma$, which are constant, and design parameter $q$. These parameters are expected to be constant for long periods of time and hence frequent re-encoding is not necessary.


At the end of the computing phase, each computed IV $\cv_i\mathbf{X}$ is available at $r_i$ ENs, where $r_i$ can be shown to lie in the interval $[r_{min}, r_{max}]$, with $r_{min}=\max\{\rho_2-(K-q),1\}$ and $r_{max}=\min\{q,\rho_2\}$ in a manner similar to \cite{ZS:18}. Moreover, for any $r_i\in[r_{min}, r_{max}]$, the number of computed IVs is $B_i=\binom{q}{r_i}\binom{K-q}{\rho_2-r_i}b$ since there are $\binom{q}{r_i}\binom{K-q}{\rho_2-r_i}$ subsets of ENs that have computed the same IVs. For downlink transmission, in order to maximizing cooperative opportunities, the computed IVs are sent in descending order of redundancy $r_i$ by using cooperative ZF precoding to serve $r_i$ users simultaneously. 


\begin{proposition}
With the described hybrid scheme, the average total latency \eqref{def:total} is given as 
\begin{align}  \label{latencyhs} 
\delta_{HS}&=\min_{q}\Bigg[ \frac{(H_K-H_{K-q})}{\eta\tau}+m\mu  \notag \\
& +\gamma \min_{(\rho_1,\rho_2)} \bigg(\sum_{r_i=r_q}^{r_{max}}\frac{B_i}{r_i}+\frac{m-\sum_{r_i=r_q}^{r_{max}} B_i}{r_q-1}\bigg)\Bigg],
\end{align}
where we have defined $r_q= \inf\big\{r:\sum_{r_i=r}^{r_{max}}B_i\leq m\big\}$; and the optimization over parameters $q\in[\lceil 1/\mu \rceil, K]$, $\rho_1 \in[1,(q+1)/q,\cdots,K/q]$, and $\rho_2\in[\lfloor q\mu \rfloor: \lfloor K\mu \rfloor]$ is constrained by Condition \eqref{condition1} and \eqref{condition2}.
\end{proposition}

\begin{proof}
Given any design parameter $q\in[\lceil 1/\mu \rceil, K]$, the average computation time is evaluated as in Proposition~\ref{pro:mds}, with the computing latency given as $(H_K-H_{K-q})/(\eta\tau)+m\mu$ in \eqref{latencymc}. Using downlink transmission, the $B_i$ IVs with redundancy $r_i$ require a communication latency $B_i/r_i$ using cooperative ZF as explained in Section~\ref{uc}. In order to deliver $m$ IVs, the IVs with redundancy $r_i\in[r_q, r_{max}]$ are sent in full, while only $m-\sum_{i=r_q}^{r_{max}}B_i$ IVs with redundancy $r_q-1$ need to be delivered. The corresponding total communication latency is optimized over all design parameters $(q, \rho_1, \rho_2)$ that satisfy Condition \eqref{condition1} and \eqref{condition2}. 
\end{proof}

\begin{figure}[t!] 
  \centering
\includegraphics[width=0.78\columnwidth]{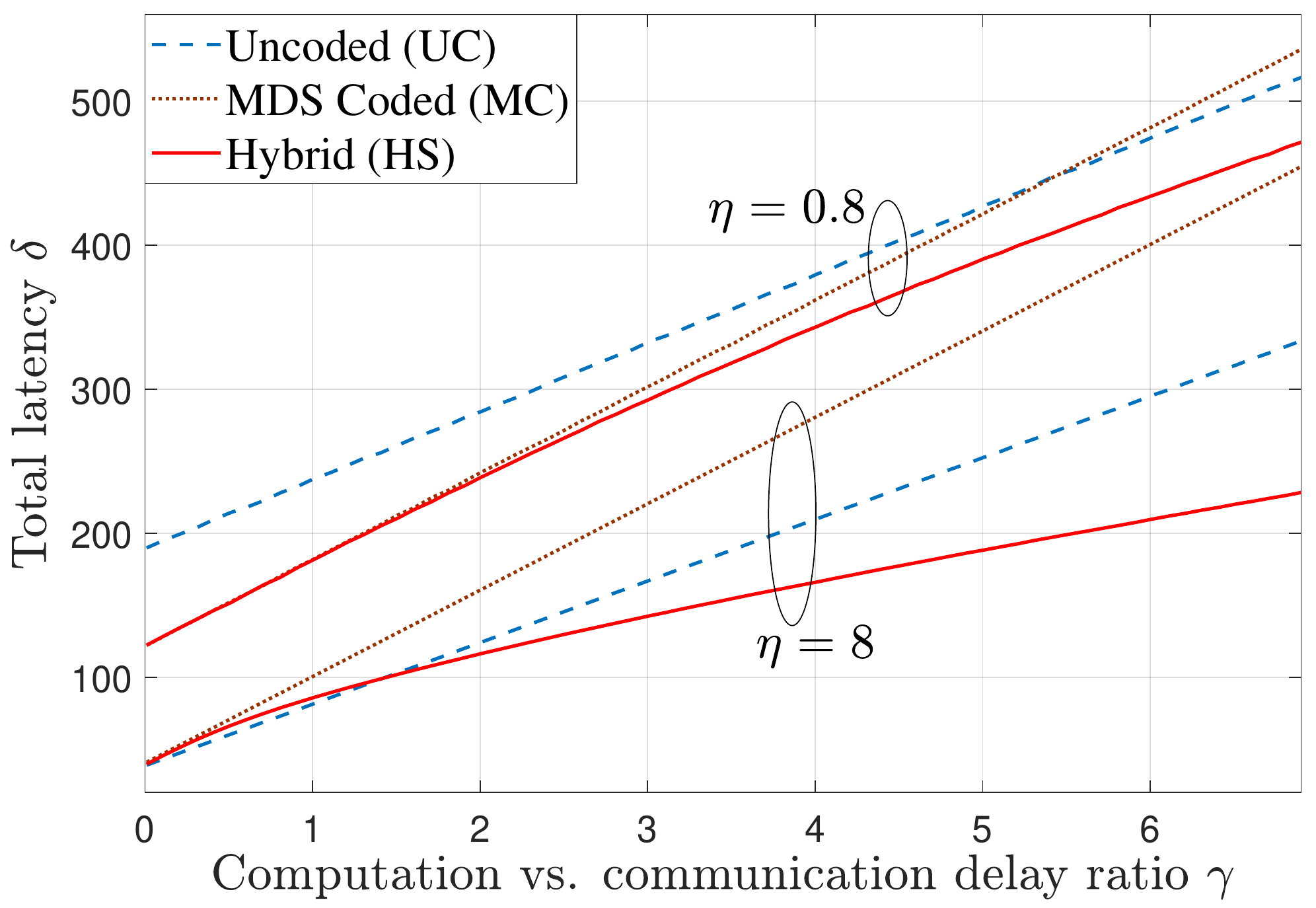}
\vspace{-9pt}
\caption{Latencies of UC, MC and HS versus ratio $r$ for $K=N=6$, $\tau=0.005$, $m=60$, $\mu=0.5$ and different values of $\eta$.}
\label{fig:delta1}
\vspace{-16pt}
\end{figure}

\vspace{-2pt}
\section{Example and Discussion}
In this section, we present a numerical example for a system with $K=N=6$ ENs and users, $m=60$ row vectors in model matrix $\mathbf{W}$, and fractional cache size $\mu=0.5$. We also set the per-IV computation time to $\tau=0.005$ and the average set-up time to different values of $1/\eta$. In Fig.~\ref{fig:delta1}, we plot the overall average latency $\delta$ as a function of the ratio $\gamma$ between normalized computation and communication times.  


As seen in Fig.~\ref{fig:delta1}, as $\gamma$ increases, the total latencies of both UC in \eqref{latencyuc} and MC in \eqref{latencymc} grow linearly, and the relative performance depends on the values of $\gamma$ and $\eta$. When $\eta$ is small, i.e., $\eta=0.8$, the variability in the computing times of the ENs is high, and MDS coding for the most part outperforms the UC scheme due to its robustness to stragglers. This is unless $\gamma$ is large enough, in which downlink transmission latency becomes dominant and the UC scheme can benefit from redundant computations via cooperative EN communication. In contrast, for larger values of $\eta$, the computing times have low variability and MDS coding is uniformly outperformed by the UC scheme.

We also observe that the proposed hybrid coding strategy is effective in trading off computation and communication latencies by controlling the balance between robustness to stragglers and cooperative opportunities via the design of parameters $(q, \rho_1,\rho_2)$. In fact, by increasing $q$ and $\rho_2$, this approach can decrease the communication latency at the cost of a larger computing latency. Apart from very small values of $\gamma$ for large $\eta$, the scheme is seem to outperform both MDS and UC strategies. 

An interesting open problem is to design a hybrid strategy that generalizes both the proposed MDS and UC schemes by properly optimizing the scheduling matrix in a manner akin to UC. Other aspects that are left for future work include the investigation of coding schemes that enable the use of ENs' partial computations \cite{MCJ:18}; of transmission strategies that carry out simultaneous edge computing and downlink communications; of the impact of partial uplink connectivity; and of protocols able to accommodate an arbitrary number of computing tasks.

\bibliographystyle{IEEEtran}
\bibliography{IEEEabrv,final_refs}

\end{document}